%% file: main.tex
\documentclass[aps,prl,twocolumn,longbibliography,superscriptaddress,notitlepage,floatfix]{revtex4-2}

\usepackage{amsmath,amsfonts,amssymb,color,epsfig,graphics,graphicx,latexsym,theorem,url,multirow,diagbox}
\usepackage[colorlinks,colorlinks,citecolor=blue,linkcolor=blue,urlcolor=black]{hyperref}
\usepackage[utf8]{inputenc}
\usepackage{booktabs}
\usepackage[T1]{fontenc}
\usepackage{courier}
\usepackage{listings}
\usepackage{latexsym}
\usepackage{dcolumn}
\usepackage{comment}
\usepackage{times} 
\usepackage{algorithm}
\usepackage{algorithmic} 
\usepackage{braket}
\usepackage{bm}
\usepackage{lineno}
\usepackage{siunitx}
\usepackage{mathtools}
\usepackage[sectionbib]{chapterbib}  

\usepackage{tikz}
\usepackage{tikz-qtree}
\usepackage[qm]{qcircuit}
\newtheorem{proof}{Proof}

\newtheorem{theorem}{Theorem}

\newcommand{\beginsupplement}{%
	\setcounter{table}{0}
	\renewcommand{\thetable}{S\arabic{table}}%
	\setcounter{figure}{0}
	\renewcommand{\thefigure}{S\arabic{figure}}%
	\setcounter{equation}{0}
	\renewcommand{\theequation}{S\arabic{equation}}%
	\setcounter{section}{0}
	\renewcommand{\thesection}{\arabic{section}}%
        \setcounter{page}{1}
        \renewcommand{\thepage}{S\arabic{page}}
}

\begin{document}
\title{Simulating fluid vortex interactions on a superconducting quantum processor}
\affiliation{Department of Engineering Mechanics, School of Aeronautics and Astronautics, Zhejiang University, Hangzhou, 310027, China\\
$^{2}$School of Physics, ZJU-Hangzhou Global Scientific and Technological Innovation Center,\\ and Zhejiang Key Laboratory of Micro-nano Quantum Chips and Quantum Control, Zhejiang University, Hangzhou 310027, China\\
$^{3}$State Key Laboratory for Turbulence and Complex Systems, College of Engineering, Peking University, and HEDPS-CAPT, Peking University, Beijing, 100871, China}

\author{Ziteng Wang$^{1}$}\thanks{These authors contributed equally to this work.}
\author{Jiarun Zhong$^{2}$}\thanks{These authors contributed equally to this work.}
\author{Ke Wang$^{2}$}
\author{Zitian Zhu$^{2}$}
\author{Zehang Bao$^{2}$}
\author{Chenjia Zhu$^{1}$}
\author{Wenwen Zhao$^{1}$}
\author{Yaomin Zhao$^{3}$}
\author{Yue Yang$^{3}$}
\author{Chao Song$^{2}$}\email{chaosong@zju.edu.cn}
\author{Shiying Xiong$^{1}$}\email{shiying.xiong@zju.edu.cn}

\begin{abstract}
\textbf{Vortex interactions are commonly observed in atmospheric turbulence, plasma dynamics, and collective behaviors in biological systems.
However, accurately simulating these complex interactions is highly challenging due to the need to capture fine-scale details over extended timescales, which places computational burdens on traditional methods.
In this study, we introduce a quantum vortex method, reformulating the Navier--Stokes (NS) equations within a quantum mechanical framework to enable the simulation of multi-vortex interactions on a quantum computer. We construct the effective Hamiltonian for the vortex system and implement a spatiotemporal evolution circuit to simulate its dynamics over prolonged periods. By leveraging eight qubits on a superconducting quantum processor with gate fidelities of 99.97\% for single-qubit gates and 99.76\% for two-qubit gates, we successfully reproduce natural vortex interactions. This method bridges classical fluid dynamics and quantum computing, offering a novel computational platform for studying vortex dynamics. Our results demonstrate the potential of quantum computing to tackle longstanding challenges in fluid dynamics and broaden applications across both natural and engineering systems.}
\end{abstract}

\maketitle
\thispagestyle{empty}

\begin{figure*}
\centering
\includegraphics[width=15cm]{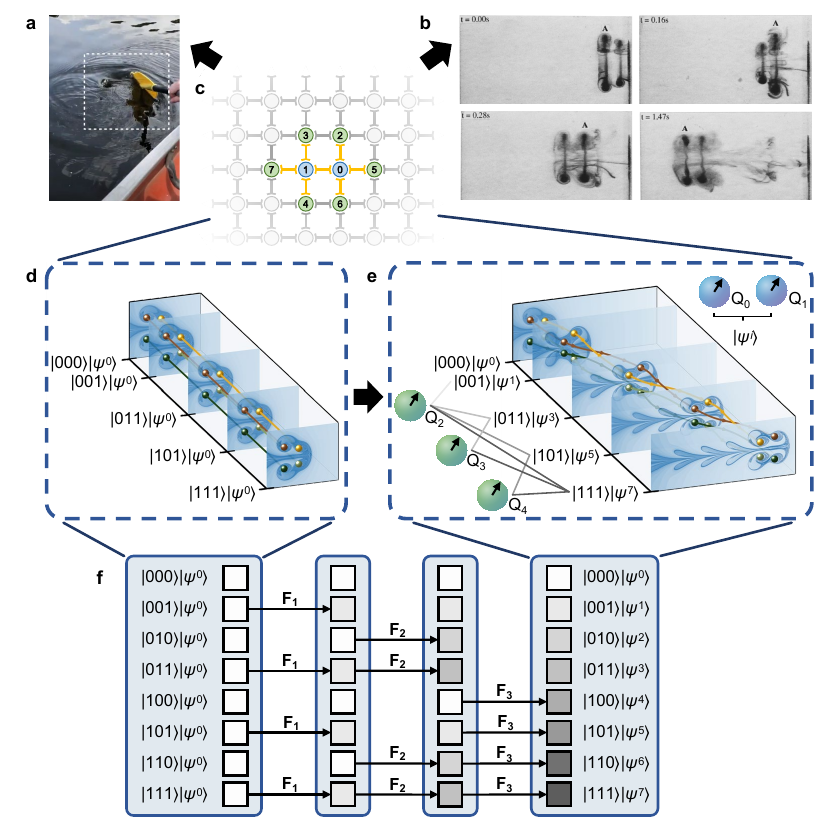}
\caption{\textbf{Overview for implementing vortex interactions
using a superconducting quantum chip.} 
\textbf{a,} Vortex pairs generated by paddling in natural fluid systems. 
\textbf{b,} Laboratory-induced vortex interactions, leading to a frog-leap configuration (Reprinted from Lim (1997)~\cite{lim1997note}, with the permission of AIP Publishing.). 
\textbf{c,} Schematic of our superconducting quantum chip, where all qubits are arranged in a square lattice with nearest-neighbor couplings.
Blue and green circles represent the spatial and temporal qubits used in our experiment.
\textbf{d,} 
The initial state of a four-vortex-particle system, which can be expressed as the tensor product of $|\Psi^0\rangle$ that encodes the initial spatial information of the vortex particles and a uniform superposition state that encodes the temporal information.
The quantum state is prepared with $n_p$ spatial qubits and $n_t$ temporal qubits.
Here, we take \(n_p = 2\) and \(n_t = 3\) for illustration.
\textbf{e,} System's final state, where each basis state of the temporal qubits is entangled with the corresponding spatial state, $|\Psi^i\rangle$, with $i$ indexes the time step.
\textbf{f,} The quantum parallel evolution enabled by our circuit. \(\bm{F}_k\) represents the unitary that evolves the system for $2^{k-1}$ time steps, with \(k\) ranging from $1$ to $n_t$. 
}
\label{fig:overview}
\end{figure*}
\vspace{.5cm}
\noindent\textbf{\large{}INTRODUCTION}

\noindent Vortices in fluids constitute a core component of complex flow behaviors, encompassing phenomena such as tropical cyclones~\cite{fiorino1989some,rios2024review,wang2004current}, ocean currents~\cite{mcwilliams1985submesoscale,zhang2024three,jebri2024absence}, microfluidics~\cite{han2022microfluidic,thurgood2023dynamic}, as well as plasmas and magnetofluids~\cite{tanaka2019vortex,izhovkina2016interaction,xiong2020effects,Aluie2009,belashov2017interaction}. Vortex interactions, which involve complex behaviors like vortex pairing and the leapfrogging effect (as shown in Fig.~\ref{fig:overview}a and Fig.~\ref{fig:overview}b), affect energy transport, momentum exchange, and the scale cascade process in fluids, ultimately determining turbulence characteristics and its evolution~\cite{shen2024designing,almoguera2024vortex,majda2002vorticity,morton1984generation}. However, simulating these critical and intricate structures using classical computation is highly challenging, as achieving the necessary spatial and temporal resolution to capture fine-scale details over extended timescales demands massive computational resources~\cite{ishihara2009study,ferziger2019computational,kochkov2021machine,blocken2015computational}, often exceeding practical limits.  This complexity has spurred the development of advanced methods to address the computational bottlenecks while maintaining accuracy.

Recent progress in quantum computing presents a promising avenue to address these challenges, as emerging research on universal quantum partial differential equation (PDE) / ordinary differential equation (ODE) solvers\cite{liu2021efficient, tennie2024solving, lloyd2020quantum, pfeffer2022hybrid, kyriienko2021solving, lubasch2020variational, harrow2009quantum, childs2017quantum} demonstrates potential for application in computational fluid dynamics (CFD) by leveraging quantum algorithms to replace key components of traditional solvers based on the fluid governing equations\cite{steijl2018parallel, gaitan2020finding, chen2022quantum, lapworth2022hybrid, demirdjian2022variational, succi2024ensemble, liu2023quantum, jaksch2023variational}. Additionally, alternative fluid dynamics descriptions optimized for quantum computing have been proposed, including quantum algorithms inspired by the lattice Boltzmann method~\cite{ljubomir2022quantum,itani2024quantum,sanavio2024lattice,budinski2021quantum,wang2025quantum}, quantum simulations based on Schr\"odingerization~\cite{jin2024quantum,jin2023quantum,jin2023analog}, and the hydrodynamic Schr\"odinger equation~\cite{meng2023quantum,meng2024simulating}, which is inherently more suitable for quantum computing than the conventional Navier--Stokes (NS) equations~\cite{succi2023quantum}.

Although quantum computing has demonstrated its potential in fluid mechanics, simulating fluid motion on actual quantum devices based on existing algorithms remains challenging. 
Most current implementations are confined to relatively simple scenarios, with complex phenomena such as vortex interactions still largely unexplored. 
This limitation is primarily due to the fact that the majority of current algorithms rely on Eulerian methods, which require high spatial resolution to accurately capture fluid behaviors, significantly increasing the quantum resources needed, as qubit requirements grow with grid resolution. 
Moreover, many quantum algorithms for simulating the time evolution of a system typically require a measurement at every time step to extract information necessary for studying dynamical behavior, computing physical quantities, and optimizing algorithms. 
However, since measurements collapse the quantum state, the quantum state at intermediate steps must be re-prepared to continue the computation. 

In this study, we reformulate the classical vortex method into a framework compatible with quantum computation and propose a novel spatiotemporal encoding scheme that embeds both spatial and temporal information directly into the initial quantum state. 
This approach enables the retrieval of information at multiple time steps from a single quantum execution, thereby eliminating the need for stepwise state preparation. 
Specifically, we propose the quantum vortex method (QVM), which directly focuses on vortices themselves instead of relying on spatial discretization grids as in the Eulerian methods~\cite{ferziger2002computational,versteeg2007introduction,zienkiewicz2005finite}, thereby enabling the reformulation of complex vortex interactions in fluids within the framework of quantum computing.
The QVM transforms the evolution of the vortex particle system into the evolution of a wavefunction. We adopt a data-driven strategy to train an evolution module that captures the dynamics of the wavefunction. 
Leveraging this trained module, we then design an efficient spatiotemporal evolution circuit to implement the wavefunction propagation, where spatial qubits encode the spatial information of the vortex particle system, while auxiliary temporal qubits, initialized into a superposition state via Hadamard gates, act as placeholders for all time steps and later serve as control qubits to guide the evolution module in the spatial circuit. 
This design effectively eliminates the necessity of performing measurements at every time step for state retrieval and avoids the repeated state preparations that are otherwise required due to measurement-induced collapse.
Building upon these theoretical developments, we implement the QVM on superconducting quantum processors to efficiently compute vortex interaction dynamics.
This approach bridges classical fluid dynamics and quantum simulations, providing a new platform for exploring both quantum and classical vortex phenomena and offering a powerful tool for reinterpreting classical vortex dynamics from a quantum perspective. 

\vspace{.5cm}
\noindent\textbf{\large{}RESULTS}

\noindent\textbf{\large{}Quantum vortex method}

\noindent The fluid dynamics are governed by the NS equations for the velocity field \(\bm{u}(\bm{x}, t)\), which describe the evolution of the flow under the influence of pressure \(p\), viscosity \(\nu\), and external forces \(\bm{f}\):
\begin{equation}
\begin{dcases}
\frac{\operatorname{D} \bm{u}}{\operatorname{D} t} = -\frac{1}{\rho} \bm \nabla p + \nu \nabla^2 \bm{u} + \bm{f}, \\
\bm{\nabla} \cdot \bm{u} = 0,
\end{dcases}
\label{eq:NS}
\end{equation}
where \(\operatorname{D}/\operatorname{D}t = \partial / \partial t + \bm{u} \cdot \bm \nabla\) is the material derivative, and \(\rho\) is the constant fluid density. 

To adapt the NS equations for quantum computing, we utilize the relationship between the vorticity field \( \bm{\omega} \) and the velocity field \( \bm{u} \) (\( \bm{\omega} = \bm{\nabla} \times \bm{u} \)), discretize the vorticity field into \( N_p \) point vortices, and map their coordinates to complex variables, leading to the generalized Schrödinger equation:
\begin{equation} \label{eq:psi_H}
\frac{\operatorname{d} \psi_j}{\operatorname{d} t} = \operatorname{i} H(\psi_1, \dots, \psi_{N_p}),
\end{equation}
where \( H \) is a vortex-interaction-dependent effective Hamiltonian.
Additionally, the wave function \(\psi_j\) is transformed from the \(j\)-th vortex particle position \(\phi_j\) with 
\begin{equation}
\psi_j = \lambda \left[\phi_j - \left( \int_0^t c(\tau) d\tau + c_0 \right) \right],
\label{eq:encoding}
\end{equation}
where \(c_0\) is an arbitrary constant, \(j\) indexes the vortex particles, and \(\lambda\) is a scaling factor that ensures the normalization condition:
$
\sum_{j=1}^{N_p} |\psi_j|^2 = 1 ~ \text{at} ~ t = 0.$
The time-dependent function \(c(t)\) is defined as:
\begin{equation}
c(t) = \operatorname{i} \lambda \frac{1}{4\pi} \sum_{j \neq k, j = 1, k = 1}^{N_p} \Gamma_k \frac{\psi_j \overline{\psi}_k - \psi_k \overline{\psi}_j}{|\psi_j - \psi_k|^2 \left( \sum_i^{N_p} \overline{\psi}_i \right)},
\label{eq:ct}
\end{equation}
Here \(\Gamma_k\) denotes the vortex strength of the \(k\)-th vortice.

This quantum formulation enables the efficient simulation of vortex dynamics on quantum hardware. 
The evolution of the quantum vortex system is governed by equations that ensure the normalization of the quantum state, facilitating accurate simulations of fluid flows with reduced computational costs. 
Furthermore, we observe that when the vortex particle system exhibits collective motion in a certain direction, $c(t)$ tends to remain relatively stable, exhibiting only minor fluctuations around a constant value. 
Therefore, we also provide a random sampling approximation method, in which we randomly select a subset of time instances of $c(t)$ and average them to approximate their complete set. A detailed description of the QVM is in SUPPLEMENTARY NOTE 1.

\vspace{.5cm}
\noindent\textbf{\large{}Quantum encoding and evolution}

\noindent The fluid dynamics are governed by Eq.~\eqref{eq:NS}, with vortex particle positions represented by~\(\phi\), which can be transformed into wave function~\(\psi\) through Eq.~\eqref{eq:encoding}, and \(\psi\) evolves according to Eq.~\eqref{eq:psi_H}. To generate a dataset for \(\psi\), one must solve Eq.~\eqref{eq:NS} on a grid, extract \(\phi\), and apply the transformation to obtain \(\psi\), thus creating the data needed for training the nonlinear model described by Eq.~\eqref{eq:psi_H}.

We investigate a system governed by Eq.~\eqref{eq:psi_H} with $N_p$ vortices, discretizing time evolution into evenly spaced $N_t$ intervals. For the \(j\)-th vortex particle, its position in configuration space is mapped to a complex variable \(\phi_j\), with $j$ ranging from 1 to $N_p$.
Subsequently, we introduce a transformation that shifts and scales the complex coordinates \(\phi_j\) to define new variables \(\psi_j\). 
This transformation ensures that each component \(\psi_j\) of the wave function is properly normalized in terms of probability and remains conserved during the evolution governed by QVM.
In the case of time discretization, the value of \(\psi_j\) at the \(i\)-th time step is denoted by \(\psi_j^i\). 
At each time step, the collection of \(\psi_j^i\) forms a state vector \( |\bm{\psi}^i\rangle = \begin{bmatrix} \psi^i_1, \psi^i_2, \ldots, \psi^i_{N_p} \end{bmatrix}^\mathrm{T}.\)

To achieve efficient evolution, the initial flow field state \(|\bm{\psi}^0\rangle\) is first encoded into a larger quantum system. 
Specifically, the system's initial state is prepared as a tensor product of the flow field's initial state, which is encoded in \(n_p=\lceil\log_2N_p\rceil\) qubits, and a uniform superposition state encoded in \(n_t=\lceil\log_2N_t\rceil\) qubits.
This procedure effectively generates multiple replicas of the flow field's initial state, as shown in Fig.~\ref{fig:overview}d.
These replicas explore different temporal evolutions simultaneously. The replicas evolve in a branching manner as depicted in Fig.~\ref{fig:overview}f, eventually yielding a superposition of flow field states at different time steps
\( |\bm{\psi}\rangle = \frac{1}{\sqrt{N_t}} \begin{bmatrix}|\bm{\psi}^0\rangle, |\bm{\psi}^1\rangle, \ldots, |\bm{\psi}^{N_t-1}\rangle \end{bmatrix}^\mathrm{T},\) as shown in Fig.~\ref{fig:overview}e.

\begin{figure}
\centering
\includegraphics[width=6cm]{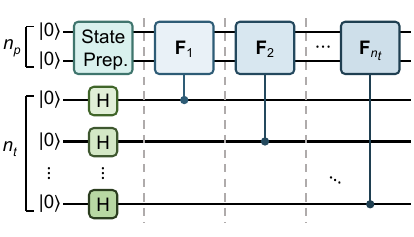} 
\caption{\textbf{The evolution circuit of the QVM method.} The circuit consists of \(n_p\) spatial qubits encoding the spatial state of vortex particles and \(n_t\) temporal qubits encoding the temporal information. The spatial qubits are initialized via the "State Prep." module, while the temporal qubits are prepared in a uniform superposition state using Hadamard gates.
The quantum parallel evolution illustrated in Fig. \ref{fig:overview}\textbf{f} is achieved through sequentially applying the controlled-\(\bm{F}_k\) operations to the system, with the temporal and spatial qubits being the control and target, respectively.
}
\label{fig:circuit}
\end{figure}

The evolution process is implemented with the quantum circuit illustrated in Fig.~\ref{fig:circuit}a. 
The core element of this circuit is the evolution unitary \(\bm{F}_k\), with \(k\) ranging from 1 to $n_t$, which evolves the state from \(i\)-th time step to \((i+2^{k-1})\)-th time step as \( |\bm{\psi}^{i+2^{k-1}}\rangle = \bm{F}_k\, |\bm{\psi}^i\rangle. \)
At the beginning of the evolution, all qubits are initialized in the state \(\ket{0}\). The system then undergoes spatiotemporal evolution through a layered quantum circuit architecture, evolving the quantum state into \(\ket{\bm \psi}\).
Specifically, the spatial qubits are initialized to the desired initial state through the ``State Prep.'' module, while the temporal qubits are prepared into a uniform superposition state via Hadamard gates. 
The temporal qubits then act as control qubits, which sequentially control the implementation of evolution unitaries \(\bm{F}_{1}, \bm{F}_{2}, \dots, \bm{F}_{k}\) on the spatial qubits.
With this controlled-unitary scheme, the quantum state undergoes a tree-like branching evolution from the initial state shown in Fig.~\ref{fig:overview}f, ultimately resulting in a superposition of all system states across $2^{n_t}$ time steps. 
This design fully leverages the parallelism of quantum computing, significantly improving the efficiency of the simulation. 
See SUPPLEMENTARY NOTE 2 for more details.

\vspace{.5cm}
\noindent\textbf{\large{}Experimental setup}

\noindent The algorithm is implemented with eight frequency-tunable transmon qubits on a flip-chip superconducting quantum processor, as shown in Fig.~\ref{fig:overview}c, where blue circles represent the spatial qubits and green circles represent the temporal qupits. Each qubit can be controlled and readout individually. The nearest-neighbouring qubits are connected with a tunable coupler for tuning on and off the effective coupling strength of the two qubits. The single qubit gate with a length of 24 ns is realized by appling a Gaussian-shape microwave pulse with DRAG correction \cite{PhysRevLett.112.240504}. The two-qubit CZ gate, with a duration of around 40 ns, is realized by tuning the frequencies and coupling strength of the two qubits to achieve a close-cycle diabatic transition, a global process involving both qubits as a combined system, between \(\ket{11}\) (both qubits in their excited states) and \(\ket{02}\) (the first qubit in the ground state and the second qubit in its second excited state), accumulating a \(\pi\) phase shift that transforms \(\ket{11}\) into \(-\ket{11}\). The median parallel single-qubit gate and two-qubit gate fidelities are 99.97\% and 99.76\% respectively. See SUPPLEMENTARY NOTE 4 for details.

\vspace{.5cm}
\noindent\textbf{\large{}Nonlinear interactions in vortex systems}

\noindent The leapfrog vortex phenomenon \cite{Aref1983}, described by the NS equation, refers to a mode of motion that occurs when two or more vortex rings interact, corresponding to four or more vortex particles in two dimensions. 
We consider the evolution of a four-vortex-particle system under Eq.~\eqref{eq:psi_H}, with the positions of the four particles initialized to be $ (0, 1), (0, 0.3), (0, -1)$, and $(0, -0.3)$, respectively.
We then apply the transformation defined in Eq.~\eqref{eq:encoding}, with $c_0 = -1.7903$ and $\Gamma = (1, 1, -1, -1)$ to obtain the corresponding quantum state. 
To learn the Hamiltonian of the vortex system, we select $100$ vortex state pairs at time $(t_i, t_i+1)$ to form the training set.
Here, $t_i= 0.18(i-1)$ with $i=1,\dots,100$ is equally sampled from a time range of $[0,18]$, which roughly corresponds to the period of a full leapfrogging cycle.
In our experiment, we use two qubits to encode spatial information of the four vortex particles.
Additionally, six qubits are used to represent $64$ time steps in the evolution. 
We then apply the QVM circuit based on the learned Hamiltonian to prepare the quantum state that encodes the spatiotemporal dynamics of the entire evolution process. 
To obtain the evolution trajectories of the vortex particles, we perform quantum state tomography (QST) on the two spatial qubits, while simultaneously conducting projective measurements on all temporal qubits.
For each time step, we postselect the QST data for the respective temporal state to reconstruct the density matrix, from which the positions of the vortex particles can be extracted (see Methods).
Note that additional global phases, which are experimentally unobservable, are numerically applied at each time step to preserve the symmetry of the system and constrain particle motion along the positive real axis.
For comparison, we conduct ideal (noiseless) and noisy simulations using the same circuit as in the experiments.
In the noisy simulation, we consider error models including depolarizing error of the single- and two-qubit gates, the qubit decoherence, and readout error, with the error rates obtained from experiments (see SUPPLEMENTARY NOTE 4).
In Fig.~\ref{fig:result1}a, we plot the experimentally extracted trajectory of the four vortex particles for time steps outside of the training set, i.e., after $t=18$. 
The results demonstrate a decent agreement between experimental data and noisy simulation.
To characterize the experimental performance, we compare the reconstructed state of the spatial qubits and positions of the vortex particles with those obtained from noiseless numerical simulation (Fig.~\ref{fig:result1}b).
The state fidelity values exceed $97\%$ for all time steps (Fig.~\ref{fig:result1}b, upper panel).
Besides, the position deviations of all four vortex particles from the noiseless simulated values remain below 0.2 throughout the evolution (Fig.~\ref{fig:result1}b, lower panel).
Moreover, using the vortex particle position data from both experimental and noiseless simulation results, we reconstruct the velocity field for each time step based on the Biot-Savart formula and visualize it in Fig.~\ref{fig:result1}c and Fig.~\ref{fig:result1}d, respectively. 
For all illustrated time steps, the velocity fields from experimental data are in close agreement with those of the noiseless simulation, demonstrating that two vortex rings, corresponding to four point vortices in 2D, alternately pass through and move forward in a real flow field while maintaining a degree of symmetry.

\begin{figure*}[ht]
\centering
\includegraphics[width=15cm]{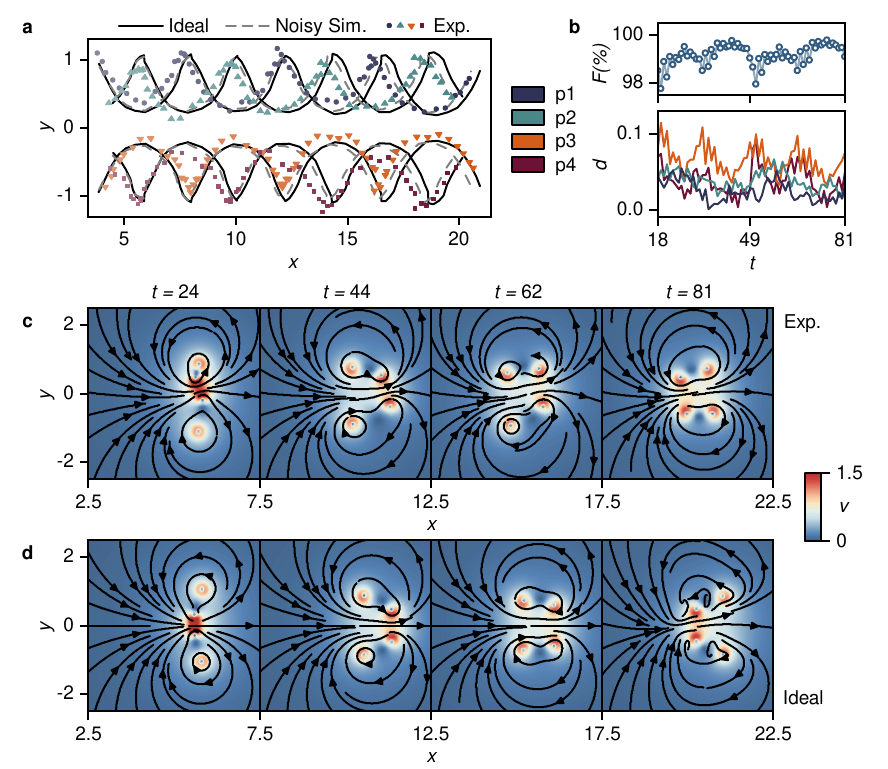}
\caption{
\textbf{Experimental results of nonlinear interactions in vortex systems.}  
\textbf{a,} Trajectories of vortex particles obtained from ideal (noiseless) simulation, noisy simulation, and experimental data.  
\textbf{b,} Fidelity and the position deviations as functions of time.  The fidelity $F$ at each time step $t$ is defined as $F=|\braket{\psi^t_\text{ideal}|\psi^t_\text{exp}}|^2$, where $\ket{\psi^t_\text{exp}}$ and $\ket{\psi^t_\text{ideal}}$ denote the state vector of the spatial qubits obtained through experiment and noiseless numerical simulation, respectively. The position deviation is defined as the Euler distance between the vortex particles in the experiment and that in noiseless simulation $d=\sum_{n=1}^{N_p}|\Vec{r}_\text{exp}^{n,t}-\Vec{r}_\text{ideal}^{n,t}|$, where $\Vec{r}_\text{exp}^{n,t}$ and $\Vec{r}_\text{ideal}^{n,t}$ denote the coordinates of the vortex particle $n$ at time $t$ obtained through experiment and noise-free simulation, respectively.
\textbf{c,d,} Velocity fields and streamlines induced by vortex particles at \(t = 24\), \(44\), \(62\), and \(81\), obtained in the experiment (\textbf{c}) and noiseless simulation (\textbf{d}).  
}
\label{fig:result1}
\end{figure*}
\vspace{.5cm}
\noindent\textbf{\large{}Turbulent vortex particle system}

\noindent To demonstrate the robustness of our method, we further implement it to simulate the dynamics of an eight-vortex-particle system.
The positions and vortex strengths of the eight vortex particles are initialized randomly, akin to a turbulent vortex particle system.
We numerically perform the simulation using MindQuantum\cite{xu2024mindspore}, an open-source quantum computing framework for simulating and implementing quantum algorithms.
We use three spatial qubits and nine temporal qubits to encode the positions and time steps, and then simulate the dynamics of the system under Eq.~\eqref{eq:psi_H}.
Specifically, we select 64 equally spaced time steps, namely, $N_t \in \{0, 4, 8, \dots, 252\}$, within the time range $[0, 256]$ as training data to directly learn the implementation circuit using the variational quantum algorithm (VQA) described in the Methods. Applying the learned circuit to the first frame, we construct the wavefunctions for all time steps from 0 to 511. 
Fig.~\ref{fig:result2}(a--d) visualizes the results at time steps 128 and 384, respectively.
The progression from (a) to (c) illustrates how vortex dynamics evolve while maintaining coherent structures, whereas the corresponding velocity distributions from (b) to (d) quantitatively capturing the spatial variations in flow velocity magnitude and direction.

\vspace{.5cm}
\noindent\textbf{\large{}Viscous vortex particle systems}

\noindent We now turn to a viscous system containing two vortex particles. For viscous vortex particle systems, our data-driven approach enables viscosity terms to be encoded within normalized quantum state vectors that preserve their physical properties during evolution. In this simple two-vortex system, both particle positions and viscous interactions can be represented through normalized wavefunctions, enabling our QVM to compute viscous evolution directly from the learned circuit using VQA on MindQuantum. In contrast, the classical Lagrangian vortex method (LVM) method faces limitations in directly incorporating viscosity terms into the vortex particle evolution.

We employ high-precision grid-based Eulerian methods for solving the NS equation to compute the two-dimensional flow field and extract vortex particle positions as the dataset. The spatial information is encoded using $n_p$ = 1 qubits, while $n_t$ = 4 qubits are allocated for encoding the temporal steps, with the first four frames used to optimize the variational circuit parameters. A comparison of the computational results between the QVM and LVM in Fig. \ref{fig:result2}e and Fig. \ref{fig:result2}f reveals that the former exhibits perfect agreement with ground truth data, whereas the latter demonstrates significant positional deviations indicative of strong viscous dissipation effects. Although 16 frames, corresponding to 4 qubits, are actually computed, only frames {0, 2, \ldots, 14} are visualized due to space limitations.

\begin{figure*}[ht]
\centering
\includegraphics[width=17cm]{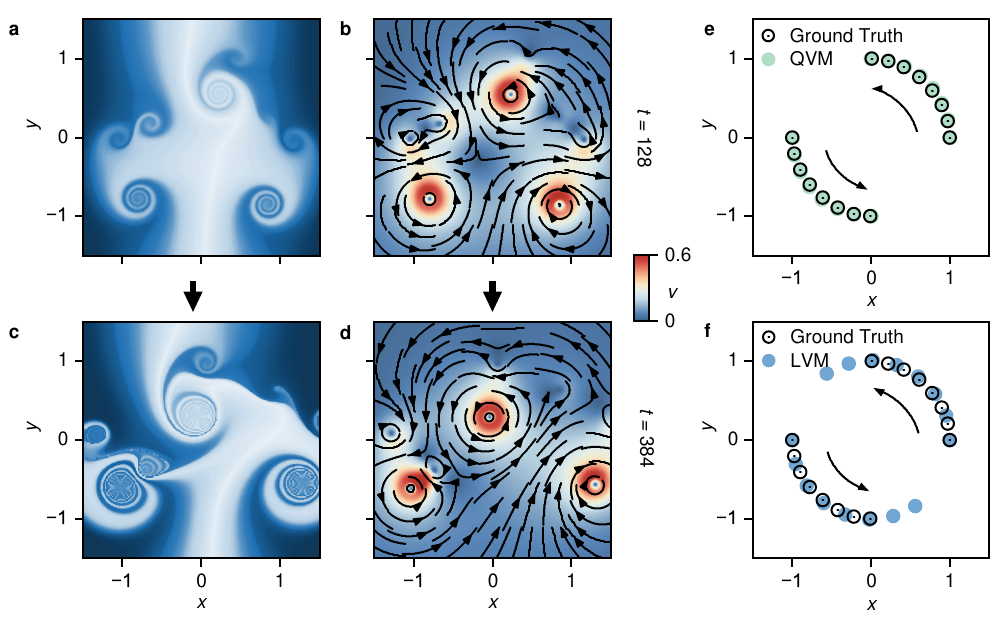}
\caption{\textbf{The simulation results of turbulent vortex particle system  and viscous vortex particle system.}
\textbf{a,} Flow field visualization rendered based on vortex particle position data from initial time to \( t = 128 \).  
\textbf{b,} Velocity distribution of the flow field at \( t = 128 \).  
\textbf{c,} Flow field visualization rendered based on vortex particle position data from the initial time to \( t = 384 \).  
\textbf{d,} Velocity distribution of the flow field at \( t = 384 \).  
\textbf{e,f} Evolution of vortex particles in a viscous fluid computed using QVM (\textbf{e}) and LVM (\textbf{f}).}  
\label{fig:result2}
\end{figure*}
\vspace{.5cm}
\noindent\textbf{\large{}DISCUSSION}

\noindent This study introduces a quantum algorithmic framework designed to simulate intricate vortex interactions in fluid dynamics. 
By directly encoding vortex information into quantum states, the approach circumvents the inherent challenges associated with quantum encoding of fluid fields. 
The proposed algorithm is validated through numerical simulations of turbulent and viscous flows, as well as experimental simulation of the leapfrogging vortex phenomenon on a superconducting quantum processor.
Our spatiotemporal encoding approach leverages both spatial and temporal degrees of freedom to dramatically expand the Hilbert space available for information storage, enabling an exponential increase in capacity compared to classical systems of similar scale. This high-density encoding scheme is particularly well-suited for storing dynamic or high-dimensional data such as neural network parameters or physical trajectories. It also aligns naturally with quantum memory architectures, allowing efficient data retrieval via quantum algorithms including Grover’s search and quantum random access memory~\cite{Ladd2010}. 
Applications span a wide range of fields, including artificial intelligence~\cite{Dean2004}, where data and models can be encoded and processed in parallel, scientific simulations involving complex many-body or time-evolving systems, and quantum cryptography~\cite{Shor1997}, where secure and scalable storage of large keys or quantum states is essential.
Future advances in measurement techniques and error mitigation strategies, such as quantum error correction~\cite{2025Quantum}, noise filtering, and more efficient tomography~\cite{Aaronson:2017qlb}, along with the development of novel quantum algorithms that reduce the need for intensive measurements, could further alleviate the computational burden and enhance the efficiency of quantum simulations.

\vspace{.7cm}
\noindent\textbf{\large{}METHODS}

\noindent\textbf{Implementing the evolution modules} \\
To implement the quantum circuit module, we leverage a data-driven approach. 
We first prepare the training data by approximating the $N_p$-particle system as a linear system described by a parameterized effective Hamiltonian $\bm{H}_{\text{eff}}(\bm{\theta})$, expressed as an $N_p \times N_p$ complex Hermitian matrix.
The training time range $T_{\text{train}}$ is uniformly divided into $N_{\text{train}}$ ($N_{\text{train}} \geq N_p^2$) segments, based on which we extract $N_{\text{train}}$ state pairs separated by a step size $\Delta t_{\text{train}}$ as our training data. While $\Delta t_{\text{train}}$ typically matches the evolution step size, it can be adjusted for specific purposes such as data interpolation.

For the training process, we construct a temporary unitary evolution operator $e^{-i\bm{H}_{\text{eff}}(\bm{\theta})\Delta t_{\text{train}}}$ based on the chosen step size $\Delta t_{\text{train}}$. This unitary matrix is not unique since $\Delta t_{\text{train}}$ can be arbitrarily selected and may differ from the actual evolution step size. 
We then optimize the parameterized unitary matrix to match the evolution for all state pairs in the training set.
The evolution matrices $\bm{F}_k$ are then constructed similarly through $e^{-i\bm{H}_{\text{eff}}(\bm{\theta})(2^k - 1)\Delta t_{\text{predict}}}$, where the time step $\Delta t_{\text{predict}}$ can be theoretically arbitrary. With the determined evolution matrices $\bm{F}_k$, the CPFlow method~\cite{Nemkov2023_Efficient} is then employed to optimize the quantum circuit design. This involves initializing multiple random parameter sets for the template circuit (ansatz) and performing parallel optimization based on the loss function. The most promising candidates are selected according to their loss values and CZ gate counts.

In addition to introducing a Hamiltonian as an intermediary to accommodate hardware limitations, the evolution module can alternatively be implemented by designing a parameterized quantum circuit ansatz and employing VQA to minimize the cost function, thereby learning vortex dynamics directly from data. We theoretically derive the complete quantum circuit for cost evaluation and gradient computation methods in SUPPLEMENTARY NOTE 3.

\vspace{.3cm}
\noindent\textbf{Extracting the spatial information} \\
In our experiment, we use QST to obtain the density matrix of the spatial qubits.
We then extract the spatial information of the vortex particles from the eigenstate of the density matrix with the maximum eigenvalue.
The method is valid for the depolarization error with a sufficiently small error rate $p$.
Under the depolarizing error channel, the experimental density matrix can be modeled as $\rho_\text{exp} \rightarrow (1-p)\rho + (p/2^N)I$, where $\rho=|\psi\rangle\langle\psi|$ is the ideal density matrix and \(N\) is the number of qubits. 
Thus, \(\ket{\psi}\) is also the eigenstate of the experimental density matrix with an eigenvalue of $1-p+p/2^N$, which remains the largest among all eigenvalues for a small $p$.
However, in our experiment, there are also coherent error channels, which can change the eigenstates of $\rho_\text{exp}$, introducing additional errors during the spatial information extraction procedure.
To mitigate this error, we apply Pauli Twirling~\cite{PhysRevA.72.052326}, which can effectively transform all errors into depolarizing errors.
Specifically, we average the QST data obtained from $50$ equivalent circuits, which are generated by randomly replacing the CZ gates in the original circuit with 16 equivalent gates realized by applying additional Pauli gates before and after the CZ gate.

\bibliographystyle{modified-apsrev4-2.bst}
\input{main.bbl}

\vspace{.3cm}
\noindent\textbf{Acknowledgement}\\
We thank Haohua Wang, Zhen Wang, Hekang Li, Qiujiang Guo, and Pengfei Zhang for their contributions to the experimental infrastructure.
The device was fabricated at the Micro-Nano Fabrication Center of Zhejiang University. 
The authors acknowledge the support of the National Key Research and Development Program of China (Grant No. 2023YFB4502600) and the National Natural Science Foundation of China (Grants No. 12302294, 12432010, and 12525201).

\vspace{.3cm}
\noindent\textbf{Author contributions}\\
S.X. conceived the theoretical ideas. J.Z. and K.W. carried out the experiment under the supervision of C.S.; Z.W. and J.Z. designed the digital quantum circuits under the supervision of C.S.; Z.W., S.X., W.Z., Y.Z. and Y.Y. conducted the theoretical analysis. J.Z., K.W., Z.Z., and Z.B. contributed to the experimental setup. All authors contributed to data analysis, discussion of the results, and writing of the manuscript.

\newpage

\clearpage
\onecolumngrid

\beginsupplement
\renewcommand{\citenumfont}[1]{S#1}
\renewcommand{\bibnumfmt}[1]{[S#1]}




\begin{center}
    \textbf{\large Supplementary Information for \\
    ``Simulating fluid vortex interactions on a superconducting quantum processors''}\\[.2cm]
    Ziteng Wang$^{1,*}$, Jiarun Zhong$^{2,*}$, Ke Wang$^{2}$, Zitian Zhu$^{2}$, Zehang Bao$^{2}$, Chenjia Zhu$^{1}$, Wenwen Zhao$^{1}$, Yaomin Zhao$^{3}$, Yue Yang$^{3}$, Chao Song$^{2,\dagger}$, and  Shiying Xiong$^{1,\ddagger}$ \\[.1cm]
    {\itshape ${}^1$ Department of Engineering Mechanics, School of Aeronautics and Astronautics, Zhejiang University, Hangzhou, 310027, China} \\
    {\itshape ${}^2$ School of Physics, ZJU-Hangzhou Global Scientific and Technological Innovation Center,\\ and Zhejiang Key Laboratory of Micro-nano Quantum Chips and Quantum Control, Zhejiang University, Hangzhou 310027, China} \\
    {\itshape ${}^3$ State Key Laboratory for Turbulence and Complex Systems, College of Engineering, Peking University, and HEDPS-CAPT, Peking University, Beijing, 100871, China} 
\end{center}

\vspace{1em}

\section*{Supplementary Note 1: quantum vortex method}
\label{QuantumVortexMethod}
We propose the quantum vortex method (QVM) by formulating the Navier–Stokes (NS) equations within a quantum framework, emphasizing key equations and their normalization conditions to facilitate the simulation of vortex dynamics using quantum algorithms. QVM adopts a particle-based approach for modeling incompressible fluid flows, is particularly well-suited for capturing vortical structures, and leverages quantum computing to potentially enhance computational efficiency.

The dynamics of an incompressible fluid are governed by the Navier--Stokes (NS) equations for velocity \(\bm{u}(\bm{x}, t)\):
\begin{equation}
\begin{dcases}
\frac{D \bm{u}}{D t} = -\frac{1}{\rho} \bm \nabla p + \nu \nabla^2 \bm{u} + \bm{f}, \\
\bm \nabla \cdot \bm{u} = 0,
\end{dcases}
\label{eq:uNS}
\end{equation}
where \(t\) is time, \(D/D t = \partial / \partial t + \bm{u} \cdot \bm \nabla\) is the material derivative, \(p\) is the pressure field, \(\nu\) is the kinematic viscosity, \(\rho\) is the constant density, and \(\bm{f}\) represents body forces. To reformulate the system in terms of vorticity, we define the vorticity field \(\bm{\omega} = \bm \nabla \times \bm{u}\), leading to an alternative representation of the NS equations:
\begin{equation}
\frac{D \bm{\omega}}{D t} = (\bm{\omega} \cdot \bm \nabla) \bm{u} + \nu \nabla^2 \bm{\omega} + \bm  \nabla \times \bm{f}.
\label{eq:wNS}
\end{equation}

We discretizes the vorticity field into a system of point vortices, with each vortex particle characterized by a position \(\bm{x}_j\) and a vortex strength \(\bm{\Gamma}_j\). The evolution of the vortex particles is described by the following system of ordinary differential equations (ODEs):
\begin{equation}
\begin{dcases}
\frac{d \bm{\Gamma}_j}{d t} = \bm{\gamma}_j, \\
\frac{d \bm{x}_j}{d t} = \bm{u}_j + \bm{v}_j,
\end{dcases}
\label{eq:vortex}
\end{equation}
where \(\bm{\Gamma}_j\) represents the integral of vorticity over the \(j^{\text{th}}\) computational element, \(\bm{u}_j\) is the induced velocity at particle \(j\), and \(\bm{v}_j\) is the drift velocity. The induced velocity \(\bm{u}_j\) is computed using Biot--Savart's law:
\begin{equation}
\bm{u}_j = \frac{1}{2(n_d-1) \pi} \sum_{k \neq j}^{N_p} \frac{\bm{\Gamma}_k \times (\bm{x}_j - \bm{x}_k)}{|\bm{x}_j - \bm{x}_k|^{n_d}},
\label{eq:BS}
\end{equation}
where \(n_d\) denotes the dimensionality of the flow and \(N_p\) is the total number of vortex particles. 

In ideal two-dimensional flow, the drift velocity \(\bm{v}_j\) and the rate of change of particle strength \(\bm{\gamma}_j\) can both be set to zero, simplifying the system further. Under these simplified conditions, the dynamical equations for the particle positions $(x_j, y_j)$ with strengths $\Gamma_j$ can be expressed in a generalized Hamiltonian \cite{xiong2020nonseparable} form. Specifically, we have
\begin{equation}
   \Gamma_j \frac{\partial x_j}{\partial t} = -\frac{\partial \mathcal{H}^p}{\partial y_j},
   \quad
   \Gamma_j \frac{\partial y_j}{\partial t} = \frac{\partial \mathcal{H}^p}{\partial x_j},
\label{eq:patialHp}
\end{equation}
where the Hamiltonian is given by
\begin{equation}
   \mathcal{H}^p
   = \frac{1}{4\pi}
     \sum_{j,k=1}^{N_p} \Gamma_j \Gamma_k
     \log\left((x_j - x_k)^2 + (y_j - y_k)^2\right).
\label{eq:Hp}
\end{equation}

We transform the vortex particle coordinates \( \bm{x}_j \) into complex variables \( \phi_j = x_j + \textrm{i} y_j \), where \( j \) indexes the vortex particles, making the representation suitable for quantum computing.
Taking the partial derivative of \(\phi\) with respect to \(t\) and combining it with \eqref{eq:patialHp}, we obtain
\begin{equation}
\frac{\partial \phi_j}{\partial t}  = \frac{1}{\Gamma_j}\left[-\frac{\partial \mathcal{H}^p}{\partial y_j}+\textrm{i} \frac{\partial \mathcal{H}^p}{\partial x_j}\right]=\frac{\textrm{i}}{\Gamma_j}\left[\frac{\partial \mathcal{H}^p}{\partial x_j}+\textrm{i}\frac{\partial \mathcal{H}^p}{\partial y_j}\right].
\end{equation}

The partial derivatives of the Hamiltonian  with respect to particle coordinates can be explicitly computed using equation \eqref{eq:Hp} as:

\begin{equation}
\frac{\partial \mathcal H^p}{\partial x_j} = \frac{1}{2\pi} \sum_{k=1}^{N_p} \Gamma_j\Gamma_k \frac{x_j-x_k}{|\phi_j-\phi_k|^2},\quad \frac{\partial \mathcal H^p}{\partial y_j} = \frac{1}{2\pi} \sum_{k=1}^{N_p} \Gamma_j\Gamma_k \frac{y_j-y_k}{|\phi_j-\phi_k|^2}.
\end{equation}

Combining these results, we find a concise complex representation:

\begin{equation}
\frac{\partial \mathcal H^p}{\partial x_j} +\textrm{i}\frac{\partial \mathcal H^p}{\partial y_j}  = \frac{1}{2\pi} \sum_{k=1}^{N_p} \Gamma_j\Gamma_k \frac{\phi_j-\phi_k}{|\phi_j-\phi_k|^2}.
\end{equation}

Thus, the evolution equation for the complex coordinates in an ideal two-dimensional flow is then expressed as:

\begin{equation}
\frac{\operatorname{d} \phi_j}{\operatorname{d} t} = \textrm{i} \frac{1}{2\pi} \sum_{\substack{k = 1 \ k \neq j}}^{N_p} \Gamma_k \frac{\phi_j - \phi_k}{|\phi_j - \phi_k|^2}.
\label{eq:phij}
\end{equation}

To prepare the system for quantum computation, we introduce a transformation that shifts and scales the complex coordinates \(\phi_j\), defining new variables \(\psi_j\):
\begin{equation}
\psi_j = \lambda \left[\phi_j - \left( \int_0^t c(\tau) d\tau + c_0 \right) \right],
\label{eq:psic}
\end{equation}
where \(c_0\) is an arbitrary constant, and \(\lambda\) is a scaling factor that ensures the normalization condition:
\begin{equation}
\sum_{j=1}^{N_p} |\psi_j|^2 = 1 \quad \text{at} \quad t = 0.
\label{eq:norm_psij}
\end{equation}

The time-dependent function \(c(t)\) is defined as:
\begin{equation}
c(t) = \operatorname{i} \lambda \frac{1}{4\pi} \sum_{j \neq k, j = 1, k = 1}^{N_p} \Gamma_k \frac{\psi_j \overline{\psi}_k - \psi_k \overline{\psi}_j}{|\psi_j - \psi_k|^2 \left( \sum_i^{N_p} \overline{\psi}_i \right)}.
\label{eq:ct}
\end{equation}

This transformation carries clear physical significance in scenarios where the fluid exhibits a pronounced directional flow. In such cases, the selected point $c_0$ can be regarded as a reference point, whose motion represents the collective movement of the vortex particle system, while the absolute motion of each vortex particle is thus redefined as motion relative to this reference point, reflected in the wavefunction as a redistribution among components under normalization constraints. One main reason for treating the motion of the reference point separately is that when the vortex particle system exhibits collective motion in a certain direction, its overall velocity remains relatively stable with only minor fluctuations around a constant, thereby enabling simplifications such as constant approximation during wavefunction-based reconstruction. This transformation enables the use of quantum algorithms for the evolution of vortex dynamics, with the complex coordinates \(\psi_j\) evolving according to
\begin{equation} \label{eq:psij1}
 \frac{\operatorname{d} \psi_j}{\operatorname{d} t}  = \operatorname{i}\frac{\lambda^2}{4\pi} \left[\sum_{k=1}^{N_p} 2\Gamma_k \frac{\psi_j-\psi_k}{|\psi_j-\psi_k|^2} - \sum_{m,n=1}^{N_p} \Gamma_n \frac{ \psi_m\overline{\psi}_n-\psi_n\overline{\psi}_m}{|\psi_m-\psi_n|^2\left(\sum_{i}^{N_p}\overline{\psi}_i\right)}\right].
\end{equation}
This quantum evolution equation represents the time evolution of the vortex system, which can be efficiently simulated using quantum computational techniques.

The transformation in equation \(\eqref{eq:psic}\) ensures the normalization of the quantum state. Specifically, the following theorem guarantees that the total probability of the system remains conserved over time.

\begin{theorem}
If \(\psi_j\) satisfies the transformation in equation \(\eqref{eq:psic}\), the normalization condition \(\sum_{j=1}^{N_p} |\psi_j|^2 = 1\) holds for all time \(t\), implying that
\begin{equation}
\frac{\operatorname{d}}{\operatorname{d} t} \sum_{j=1}^{N_p} |\psi_j|^2 = 0.
\label{eq:normlization}
\end{equation}
\end{theorem}
\begin{proof}
We multiply both sides of equation \eqref{eq:psij1} by the complex conjugate \(\overline{\psi}_j\) and take the real part to obtain
\begin{equation}
\frac{\operatorname{d} |\psi_j|^2}{\operatorname{d} t} = 2 \text{Re} \left[ \lambda^2 i \frac{1}{2\pi} \sum_{k \neq j}^{N_p} \Gamma_k \frac{\psi_j - \psi_k}{|\psi_j - \psi_k|^2} \overline{\psi}_j - \lambda c(t) \overline{\psi}_j \right].
\end{equation}
Simplifying further, we find
\begin{equation}
\frac{\operatorname{d} |\psi_j|^2}{\operatorname{d} t} =
= \lambda^2 i \frac{1}{2 \pi} \sum_{k \neq j}^{N_p} \Gamma_k \frac{\psi_j \overline{\psi}_k - \psi_k \overline{\psi}_j}{|\psi_j - \psi_k|^2} - 2 \text{Re} \left[ \lambda c(t) \overline{\psi}_j \right].
\end{equation}
Summing over \(j\) from 1 to \(N_p\), we obtain
\begin{equation}
\frac{\operatorname{d}}{\operatorname{d} t} \sum_{j=1}^{N_p} |\psi_j|^2 = \lambda^2 i \frac{1}{2 \pi} \sum_{k \neq j}^{N_p} \sum_{j=1}^{N_p} \Gamma_k \frac{\psi_j \overline{\psi}_k - \psi_k \overline{\psi}_j}{|\psi_j - \psi_k|^2} - 2 \text{Re} \left[ \lambda c(t) \sum_{j=1}^{N_p} \overline{\psi}_j \right].
\label{eq:dsumpsi}
\end{equation}
Finally, substituting the expression for \(c(t)\) from equation \eqref{eq:ct} into \eqref{eq:dsumpsi}, we arrive at the normalization condition \eqref{eq:normlization}. Thus, the normalization of \(\psi_j\) is preserved over time.
\end{proof}

We have reformulated the ODEs (\eqref{eq:vortex}) into the system described by \eqref{eq:psij1} with the normalization condition \eqref{eq:norm_psij}, which offers advantages for quantum algorithm design and simulating two-dimensional vortex dynamics. Converting from this system back to the original vortex dynamics (\eqref{eq:vortex}) requires specifying the constants \(\lambda\) and \(c_0\), as well as calculating the time-varying function \(c(t)\) from \eqref{eq:ct}. To reconstruct vortex evolution, we set \(\lambda\) and \(c_0\) based on the initial conditions, which define the problem's central position and scale. Although computing \(c(t)\) is computationally intensive and prone to integration errors, numerical experiments show that \(c(t)\) remains nearly constant, though adjustments are needed based on sampled time-evolution data.

We introduce a straightforward approximation method for reconstruction, referred to as the random sampling approximation. By randomly selecting time points according to a predetermined ratio based on the measured data, we calculate the average value of \( c(t) \). Figure \ref{fig:error-SP} illustrates the error progression for different sampling ratios in the leapfrog example. Even with a low sampling frequency, the error remains manageable. Increasing the sampling ratio can reduce the error to approximately 0.6\%, which is much lower than the measurement error in quantum computing. As a result, this approach significantly simplifies the process and reduces the computational burden of reconstructing vortex evolution.

\begin{figure}
\centering    
\includegraphics[width=8cm]{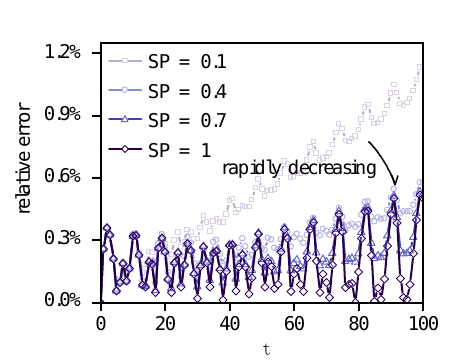}
\caption{
Time evolution of the relative error over the first 100 frames for four different sampling proportions (SP), namely 0.1, 0.4, 0.7, and 1 with initial conditions $x_1=x_2=x_3=x_4=0$, $y_1=-y_3=1$, and $y_2=-y_4=0.3$. For each sampling proportion (SP), a subset of frames is randomly selected, and the variable $c$ is computed for these frames according to Equation~\eqref{eq:ct}. The averaged value, denoted by $c_{\text{const}}$, is then used to replace the time-dependent function $c(t)$ in the evaluation of the integral defined in Equation~\eqref{eq:psic}. As the SP increases, the point-to-point error shows a rapid decreasing trend.}

\label{fig:error-SP}
\end{figure}

Figure \ref{fig:max error} compares the vortex element evolution predicted by our approximation method. With a sampling ratio of 0.4, the evolution across multiple dimensionless positions shows a max error of just 0.147, which is negligible in the context of the NISQ era of quantum computing.

\begin{figure}
\centering    
\includegraphics[width=13cm]{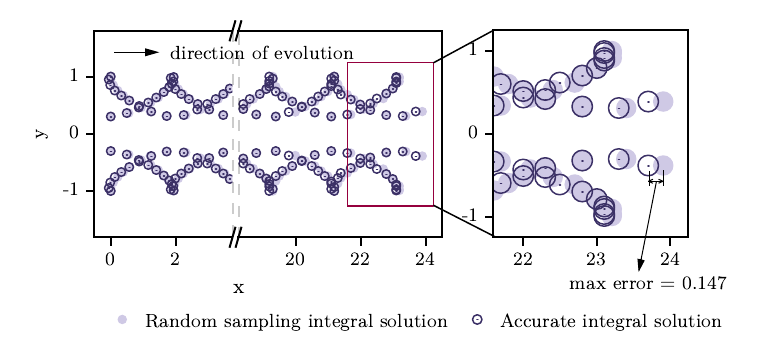}
\caption{Under a sampling proportion of 0.4, the comparative graph illustrates a maximum error of 0.147 resulting from the random sampling approximation over a displacement of 24 dimensionless units, juxtaposed with high-precision simulation results.}
\label{fig:max error}
\end{figure}

\section*{Supplementary Note 2: Construction of quantum circuits}
\label{ConstructionQuantumCircuits}

We employ quantum algorithms to solve \eqref{eq:psij1}. To achieve this, we encode the wave function $\psi_j(t)$, representing the space-time evolution of vortex particles, into quantum states. We investigate a system with $2^{n_p}=N_p$ vortices, discretizing time evolution into evenly spaced $2^{n_t}=N_t$ intervals. The variable $\psi_{j}^{i}$ represents the component of the system's wave function for a vortex particle with spatial index $j$ and temporal index $i$. Here, $j$ ranges from 1 to $N_p$, and $i$ ranges from 1 to $N_t$.

At each time point $i\Delta t$, the spatial position state of the wave function, denoted as $\ket{\bm \psi^i}$, satisfies the normalization condition and can be represented by a quantum state as follows:
\begin{equation}
\ket{\bm \psi^i} = \begin{bmatrix}
\psi^i_1\\
\vdots\\
\psi^i_{N_p}\\
\end{bmatrix}.
\end{equation}
Then we gather $\ket{\bm \psi^i}$ from all temporal instances into the quantum state $\ket{\bm \psi}$:
\begin{equation}
\ket{\bm \psi} = \frac{1}{\sqrt{N_t}} \begin{bmatrix}
\ket{\bm \psi^0}\\
\vdots\\
\ket{\bm \psi^{N_t-1}}\\
\end{bmatrix},
\label{eq:whole_psi}
\end{equation}
thereby completing the quantum encoding of the vortex particles.

Inspired by DMD methods \cite{kang2023force}, we assume a linearized temporal evolution of vortex particles represented by the mapping:
\begin{equation}
|\bm \psi^{i+1}\rangle = \bm F(\bm \theta) |\bm  \psi^i\rangle,
\label{eq:mapping}
\end{equation}
within the quantum circuit, where $\bm \theta(\theta_1,\theta_2,\cdots,\theta_L)$
denotes the variable parameters optimized by a variational algorithm, $L$ represents the total number of parameters. Starting from the initial state $|\bm{\psi_{R}}^0\rangle$, the system evolves through all states in the sequence over time without the need for repeated measurements or state preparations. The final state $|\bm{\psi}\rangle$ stores all the information of the evolution process, which can be expressed as:
\begin{equation}
|\bm \psi
\rangle=
\bm{\mathcal{U}}(\bm \theta)
|\bm \psi_{R}^0\rangle.
\end{equation}

Here, the initial state $|\bm \psi_{R}^0\rangle$ is given by:
\begin{equation}
|\bm \psi_{R}^0\rangle = \frac{1}{\sqrt{N_t}} \begin{bmatrix}
|\bm \psi^0 \rangle \\
|\bm \psi^0 \rangle \\
\vdots\\
|\bm \psi^0 \rangle \\
|\bm \psi^0 \rangle \\
\end{bmatrix}= (H|0 \rangle)^{\otimes {n_t}} \otimes |\bm \psi^0\rangle,
\label{eq:psiR}
\end{equation}
where $H$ denotes the Hadamard gate, and the wave function $\bm \psi^0$ is prepared by the circuit $\bm U_R^0$, which acts on the ground state to yield the corresponding quantum state, expressed as
\begin{equation}
    |\bm\psi^0 \rangle = \bm U_{R}^0 |\bm 0 \rangle.
    \label{eq:UR0}
\end{equation}
In experiments on quantum hardware, the circuit for preparing the initial state is obtained via the CPFlow algorithm.
Additionally, the matrix operator $\bm{\mathcal{U}}(\bm \theta)$ is given by:
\begin{equation}
\bm{\mathcal{U}}(\bm \theta)  = \begin{bmatrix}
\bm{1} & \bm{0} & \cdots & \bm{0} & \bm{0} \\
\bm{0} & \bm{F}_1(\bm{\theta}) & \cdots & \bm{0} & \bm{0} \\
\vdots & \vdots & \ddots & \vdots & \vdots \\
\bm{0} & \bm{0} & \cdots & \bm{F}_{k-1}(\bm{\theta}) & \bm{0} \\
\bm{0} & \bm{0} & \cdots & \bm{0} & \bm{F}_k(\bm{\theta}) \\
\end{bmatrix}
,
\end{equation}
where \(\bm F_k (\bm \theta)= (\bm F(\bm \theta))^{2^{k-1}}\), and $\bm F(\bm \theta)$ assumes the same operation for each timestep, serving as the fundamental module in our experiment. 
The matrix $\bm{\mathcal{U}}(\bm \theta)$ can be expressed as a product of block-diagonal matrices:

\begin{equation}
\bm{\mathcal{U}}(\bm \theta) = \prod_{k=1}^{n_t} \bm{P}_k (\bm \theta), 
\label{eq:decomposdU}
\end{equation}
where \(\bm{P}_k (\bm \theta)\) is a diagonal matrix of size \(N_t \times N_t\). Each diagonal element of \(\bm{P}_k (\bm \theta)\) is defined as follows:

\begin{equation}
(\bm{P}_k (\bm \theta))_{i,i} =
\begin{cases} 
\bm{F}_k (\bm \theta), & \text{if } (i \bmod 2^k) \geq 2^{k-1}, \\
\bm{1}, & \text{if } (i \bmod 2^k) < 2^{k-1}.
\end{cases}
\end{equation}

This matrix decomposition method ensures that $\bm{\mathcal{U}}(\bm \theta)$ comprehensively captures the system's evolution with the desired modular structure. Specifically, each $\bm{P}_k (\bm \theta)$ applies the evolution operator $\bm{F}_k (\bm \theta)$ only to the subspace of the quantum state where the $k$-th bit in its binary representation is 1. This condition-based control mechanism allows each $\bm{P}_k (\bm \theta)$ to correspond to a controlled gate operation in the quantum circuit, governed by the state of the $k$-th qubit. This hierarchical and bit-wise controlled structure efficiently constructs the matrix $\bm{\mathcal{U}}(\bm \theta)$ and predicts the states at $N_t$ future moments without repeated state preparations, resembling a tree structure as shown in Fig. \ref{fig:tree} that minimizes classical-quantum interactions.

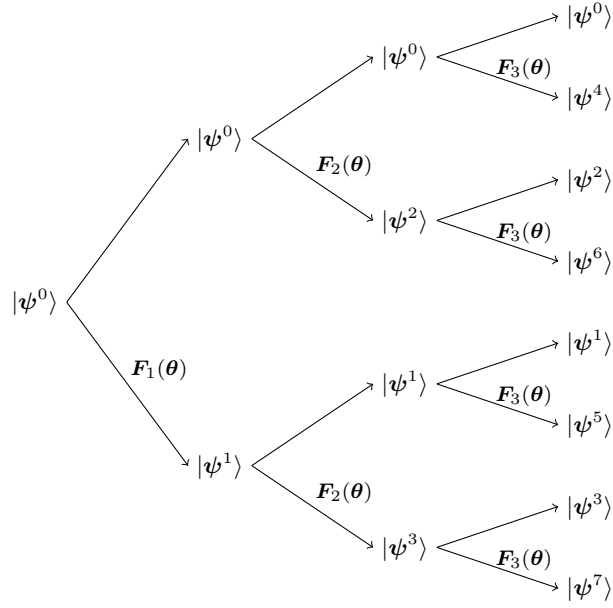
\begin{figure}
      \centering

      \begin{tikzpicture}[semithick,->,grow'=right] 
\tikzset{level distance=70pt, sibling distance=15pt} 
\tikzset{every tree node/.style={anchor=base west}}
\Tree
		[.{$|\bm \psi^0\rangle$}
			[.{$|\bm \psi^0\rangle$ }
				[.{$|\bm \psi^0\rangle$}
    				[.{$|\bm \psi^0\rangle$} ] 
                     \edge node[midway, above right,inner sep=-1pt, font=\footnotesize] {$\bm F_3 (\bm \theta)$};
    				[.{$|\bm \psi^4\rangle$} ]
                ] 
                 \edge node[midway, above right,inner sep=1pt, font=\footnotesize] {$\bm F_2 (\bm \theta)$};
				[.{$|\bm \psi^2\rangle$}
                    [.{$|\bm \psi^2\rangle$} ] 
                     \edge node[midway, above right,inner sep=-1pt, font=\footnotesize] {$\bm F_3 (\bm \theta)$};
    				[.{$|\bm \psi^6\rangle$} ]
                ]
            ]
             \edge node[midway, above right,inner sep=1pt, font=\footnotesize] {$\bm F_1 (\bm \theta)$};
			[.{$|\bm \psi^1\rangle$} 
				[.{$|\bm \psi^1\rangle$} 
    				[.{$|\bm \psi^1\rangle$} ] 
                     \edge node[midway, above right,inner sep=-1pt, font=\footnotesize] {$\bm F_3 (\bm \theta)$};
    				[.{$|\bm \psi^5\rangle$} ]
                ] 
                 \edge node[midway, above right,inner sep=1pt, font=\footnotesize] {$\bm F_2 (\bm \theta)$};
				[.{$|\bm \psi^3\rangle$} 
    				[.{$|\bm \psi^3\rangle$} ] 
                     \edge node[midway, above right,inner sep=-1pt, font=\footnotesize] {$\bm F_3 (\bm \theta)$};
    				[.{$|\bm \psi^7\rangle$} ]
                ]
			]
		]
	]
\end{tikzpicture}
      \caption{The evolution of quantum states exhibits a tree-like structure, branching from an initial state into superpositions of temporal sequences. At each step, unitary operators simultaneously act on half of the components of the quantum state, while the other half remains unchanged. As a schematic diagram, only the initial part of the tree is shown to illustrate the branching pattern.}
\label{fig:tree}
\end{figure}

Specifically, for a given value of $n_t$, the circuit depicted in Fig. \ref{temporal evolution} can be built using a decomposition as  \eqref{eq:decomposdU}.
Initially, the system state is described by \eqref{eq:psiR}. After applying the controlled quantum gate $\bm F_1 (\bm \theta)$, with the control qubit initially in the $(|0 \rangle+|1 \rangle)/\sqrt{2}$ state through the Hadamard gate, $\bm F_1 (\bm \theta)$ acts on one half of the state, leaving the other half unchanged, resulting in:
\begin{equation}
|\bm\psi_R^1  \rangle  = \frac{\sqrt{2}}{2} (H|0 \rangle)^{\otimes (n_t-1) }(|0 \rangle | \bm \psi^0 \rangle + |1 \rangle  | \bm \psi^1 \rangle) .
\end{equation}
Similarly, the state after applying the controlled quantum gate $\bm F_2 (\bm \theta)$ can be written as:
\begin{equation}
|\bm\psi_R^2  \rangle = (\frac{\sqrt{2}}{2} )^2 (H|0 \rangle)^{\otimes (n_t-2) }(|0 \rangle |0 \rangle| \bm \psi^0 \rangle  + |0 \rangle |1 \rangle | \bm \psi^1 \rangle + |1 \rangle  |0 \rangle | \bm \psi^2 \rangle+ |1 \rangle  |1 \rangle | \bm \psi^3 \rangle)
\end{equation}
After a sequence of analogous controlled gate operations, the system reaches its final state as:
\begin{equation}
\begin{aligned}
|\bm\psi \rangle=|\bm\psi_R^{n_t} \rangle
=&(\frac{\sqrt{2}}{2} )^{n_t} ( |0 \rangle  |0 \rangle...|0 \rangle | \bm \psi^0 \rangle+|0 \rangle  |0 \rangle...|1 \rangle | \bm \psi^1 \rangle+\cdots+ |1 \rangle|1 \rangle ... |1 \rangle | \bm \psi^{N_t-1} \rangle)
\end{aligned}
\end{equation}

\begin{figure}
\centering
\centerline{
\Qcircuit @C=1em @R=0.5em @!R { \\
    	 	\nghost{ } & \lstick{| 0\rangle^{^{n_p}}}& \qw  & \lstick{/} \qw  & \gate{\bm U_R^0} \barrier[0em]{5}  & \qw & \gate{\bm F_1 (\bm \theta)} \barrier[0em]{5} & \qw & \gate{\bm F_2 (\bm \theta)} \barrier[0em]{5} & \qw & \gate{\bm F_3 (\bm \theta)} \barrier[0em]{5} & \qw  & \qw &\nghost{}&\lstick{\dots \ } \nghost{}& \gate{\bm F_{n_t} (\bm \theta)} \barrier[0em]{5} & \qw  & \qw  \\
	 	\nghost{ } & \lstick{| 0\rangle } & \qw &  \qw & \gate{H} & \qw & \ctrl{-1} & \qw & \qw & \qw & \qw & \qw & \qw &\nghost{}&\lstick{\dots \ } \nghost{} & \qw&  \qw  & \qw  \\
	 	\nghost{ } & \lstick{| 0\rangle } & \qw & \qw & \gate{H} & \qw & \qw & \qw & \ctrl{-2} & \qw & \qw & \qw & \qw &\nghost{}&\lstick{\dots \ } \nghost{} & \qw & \qw   & \qw\\
	 	\nghost{ } & \lstick{| 0\rangle } & \qw & \qw & \gate{H}& \qw & \qw & \qw & \qw & \qw & \ctrl{-3} & \qw & \qw &\nghost{}&\lstick{\dots \ } \nghost{}& \qw & \qw & \qw  \\
          \nghost{ } & \lstick{\vdots\ \  } & \nghost{} & \rstick{\ \ \ \ \ \vdots} \nghost{}& \nghost{} & \nghost{} & \nghost{}& \nghost{} & \nghost{}& \nghost{}& \nghost{}& \nghost{}&\rstick{\ \ \ \ddots} \nghost{}   \\
        \nghost{ }& \lstick{| 0\rangle } & \qw  & \qw  & \gate{H} &\dstick{\ \ \ \ \ \ \ \ \ket{\bm\psi_R^0}} \qw & \qw &\dstick{\ \ \ \ \ \ \ \ \ket{\bm\psi_R^1}} \qw &  \qw &\dstick{\ \ \ \ \ \ \ \ \ket{\bm\psi_R^2}} \qw & \qw &\dstick{\ \ \ \ \ \ \ \ \ket{\bm\psi_R^3}} \qw & \qw &\nghost{}&\lstick{\dots \ } \nghost{}& \ctrl{-5} &  \qw & \dstick{\ \ \ \ \ \ \ \ \ket{\bm\psi_R^{n_t}}}\qw  \\ 
 \\  
 \\}
}
\caption{The quantum circuit isolates temporal evolution from spatial interactions by treating the spatial interactions, represented by the circuit $\bm F_k (\bm \theta)$, as a submodule of temporal evolution. The ancillary qubits become superpositions of $\ket{0}$ and $\ket{1}$ states after the Hadamard gate application. Serving as control qubits, they regulate the action of $\bm F_k (\bm \theta)$ on only half of the quantum state components.}
\label{temporal evolution}
\end{figure}
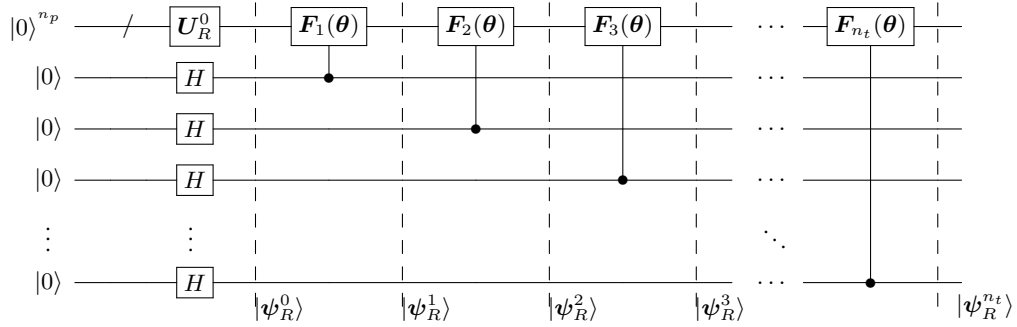

\section*{Supplementary Note 3: Implementation of Evolution Block}
\label{EvolutionBlockImplementation}
The nonlinear nature of \eqref{eq:psij1} is unsuitable for direct solution using quantum computing techniques. To circumvent this challenge, we consider an approximate linearised system. Specifically, equation \eqref{eq:psij1} is approximated by the following form
\begin{equation}
    \frac{\operatorname{d}}{\operatorname{d} t} 
   \begin{bmatrix}
\psi_1(t)\\
\vdots\\
\psi_{N_p}(t)\\
\end{bmatrix}
=-\operatorname{i} \bm{H}_{\text{eff}}(\bm{\theta})  \begin{bmatrix}
\psi_1(t)\\
\vdots\\
\psi_{N_p}(t)\\
\end{bmatrix} + \bm \varepsilon(t).
\end{equation}
Here, $\bm{H}_{\text{eff}}(\bm{\theta}) $ represents the effective Hamiltonian. Within a certain range, the original system can be approximated by a linear system. In our experiments, we determine the effective Hamiltonian for this linear approximated system by fitting a Hermitian matrix to the numerical solutions of the system's evolution over a preceding period of time.

The effective Hamiltonian $\bm{H}_{\text{eff}}(\bm{\theta}) $ is obtained by minimizing the following loss function:
\begin{equation}
\mathcal{L}(\bm{H}_{\text{eff}}(\bm{\theta}) ) = \sum_{i=0}^{M_t-1} \| |\bm \psi^{i+1}\rangle - e^{-\operatorname{i}\bm{H}_{\text{eff}}(\bm{\theta}) \Delta t }|\bm  \psi^{i}\rangle \|^2,
\end{equation}
where \(M_t-1\) denotes the needed dateset size. This loss function quantifies the deviation between the numerically evolved states and the states predicted by the effective Hamiltonian.

Considering that a Hermitian matrix is a complex symmetric matrix, its degrees of freedom are $N_p^2$ (the main diagonal contains $N_p$ real numbers, while the off-diagonal complex symmetric elements contribute $N_p \times (N_p - 1)$ degrees of freedom). Therefore, a data size of at least $N_p^2$ is required to fully determine such a matrix. Once the effective Hamiltonian $\bm{H}_{\text{eff}}(\bm{\theta}) $ is obtained, the evolution matrix $\bm{F}_k (\bm \theta)$ can be determined using:
\begin{equation}
    \bm{F}_k (\bm \theta) = e^{-\operatorname{i}\bm{H}_{\text{eff}}(\bm{\theta}) (2^k - 1)\Delta t},
\end{equation}
where $\Delta t$ represents the unit time step used in the computation process. Based on the evolution matrix, the corresponding quantum gate circuit for $\bm{F}_k (\bm \theta)$ can be constructed.

This method, which involves explicitly constructing the effective Hamiltonian, forms the basis of our experiments. However, as the system size increases, the effective Hamiltonian becomes impractical to scale due to its rapidly growing degrees of freedom. To address this limitation, we propose an alternative approach using a parameterized quantum circuit (Ansatz). In this method, the circuit parameters are optimized via a variational quantum algorithm to fit the target data, effectively bypassing the need to explicitly construct the effective Hamiltonian. 
In the following sections, we will introduce the design of the cost function and the corresponding quantum circuit, and then discuss the gradient computation techniques required for parameter optimization. These two aspects form the foundation of the proposed method, enabling it to capture the essential dynamics of the system while remaining scalable to larger system sizes.

\subsection*{A. Cost function and corresponding quantum circuit}
\label{CostFunction&QuantumCircuit}
In the optimization of ansatz parameters, the use of an appropriate cost function is crucial to gauge the proximity between predicted and ground truth values. To this end, we define the cost function using a ground truth dataset, where each data point consists of two frames separated by \( K = (2^k - 1) \) time steps. Here, \( k \) ranges from \( 1 \) to \( n_t \), and the frames are denoted as \( \bm \psi^i_{\operatorname{gt}} \) and \( \bm \psi^{i+K}_{\operatorname{gt}} \), corresponding to a total time interval of \( K\Delta t \). The generation of frames $\bm \psi^i_{\operatorname{gt}}$ and $\bm \psi^{i+K}_{\operatorname{gt}}$ involves the application of quantum circuits $\bm U_{R}^i$ and $\bm U_{R}^{i+K}$, which act on the ground state to yield the quantum states corresponding to the respective time frames:
\begin{equation}
    |\bm\psi^i_{\operatorname{gt}} \rangle = \bm U_{R}^i |\bm 0 \rangle,~~|\bm\psi^{i+K}_{\operatorname{gt}} \rangle = \bm U_{R}^{i+K} |\bm 0 \rangle.
    \label{eq:UR}
\end{equation}
The cost function is then defined as the inner product of the difference between the predicted and ground truth quantum states, resulting in a real number:
\begin{equation}
\label{eq:cost function}
C(\bm \theta)=\langle \bm \psi^{i+K}_{\operatorname{gt}} - \bm \psi^{i+K}_{\operatorname{predict}} | \bm \psi^{i+K}_{\operatorname{gt}} - \bm \psi^{i+K}_{\operatorname{predict}}\rangle .
\end{equation}
Here, $|\bm \psi^{i+K}_{\operatorname{predict}}  \rangle$ signifies the estimated result mapped by $\bm F_k(\bm \theta)$, representing the evolution over \( K \) time steps, as expressed in \eqref{eq:mapping}, namely,
\begin{equation}
|\bm \psi^{i+K}_{\operatorname{predict}} \rangle = \bm F_{k}(\bm \theta)|\bm \psi^{i}_{\operatorname{gt}} \rangle.
\end{equation}

\begin{figure}
\centering
\centerline{
    \Qcircuit @C=0.8em @R=1em @!R { 
	  \lstick{| 0\rangle^{{n_p}} :  } & \qw & \lstick{/} \qw  &  \barrier[0em]{2} \qw & \qw& \gate{\bm U_{R}^i} & \gate{\bm U_{R}^{i+1}}  &  \barrier[0em]{2} \qw &  \qw & \gate{\bm F(\bm \theta)} \barrier[0em]{2}& \qw & \qw \barrier[0em]{2} & \qw & \qw & \qw & \qw\\
	 \lstick{| 0\rangle :  } &\qw & \qw & \gate{H} & \qw& \ctrl{-1} & \ctrlo{-1} & \qw &  \qw  & \ctrl{-1} & \qw  & \gate{H}& \qw & \meter  & \qw& \qw \\
	 \lstick{\mathrm{{c} : }} & \cw & \cw & \cw  &\dstick{\ \ \ \ \ \ \ket{\bm\varphi_1}} \cw  & \cw& \cw & \cw&\dstick{\ \ \ \ \ \ \ket{\bm\varphi_2}} \cw&  \cw &\dstick{\ \ \ \ \ \ \ket{\bm\varphi_3}} \cw &   \cw &\dstick{\ \ \ \ \ \ \ket{\bm\varphi_4}} \cw & \cw \ar @{<=} [-1,0]& \dstick{_{_{\hspace{0.0em}}}} \cw   & \cw  \gategroup{1}{4}{2}{7}{.7em}{--}\\
  \\
   }}
\caption{Circuit for cost evaluation}
\label{fig:cost evaluation}
\end{figure}
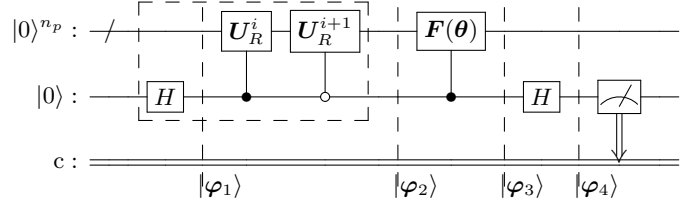

To implement \eqref{eq:cost function} in the quantum circuit, we follow a sequence of four steps, illustrated in Fig. \ref{fig:cost evaluation}. Firstly, we prepare the quantum state $\bm \varphi_1$ as:
\begin{equation}
    |\bm \varphi_1 \rangle = H|0\rangle \otimes |0\rangle^{\otimes n_p} = H  | 0\rangle | \bm 0\rangle  = \frac{1}{\sqrt{2}}(|0\rangle|\bm 0\rangle + |1\rangle|\bm 0\rangle),
\end{equation}
where $|0\rangle^{\otimes n_p}$ is abbreviated as $|\bm 0\rangle$. Successive application of the controlled $\bm U_{R}^{i+K}$ gate and $\bm U_{R}^i$ gate results in the quantum state $\bm \varphi_2$:
\begin{equation}
|\bm \varphi_2 \rangle = \frac{1}{\sqrt{2}}(|0\rangle \bm U_{R}^{i+K}|\bm 0\rangle + |1\rangle \bm U_{R}^i|\bm 0\rangle) = \frac{1}{\sqrt{2}}(|0\rangle|\bm \psi^{i+K}_{\operatorname{gt}} \rangle + |1\rangle| \bm \psi^i_{\operatorname{gt}} \rangle).
\end{equation}
Next, the system undergoes a controlled $\bm F_k(\bm \theta)$ gate, affecting only the components in the state $| \bm\psi^i_{\operatorname{gt}} \rangle|1\rangle$. The resulting quantum state $\bm \varphi_3$ is:
\begin{equation}
|\bm \varphi_3\rangle = \frac{1}{\sqrt{2}}(|0\rangle|\bm \psi^{i+K}_{\operatorname{gt}} \rangle + |1\rangle \bm F_k(\bm \theta)| \bm \psi^i_{\operatorname{gt}} \rangle) = \frac{1}{\sqrt{2}}(|0\rangle|\bm \psi^{i+K}_{\operatorname{gt}} \rangle + |1\rangle|\widetilde{\bm \psi}^{i+K}_{\operatorname{predict}} \rangle).
\end{equation}
Finally, a Hadamard gate is applied to the ancillary qubit, resulting in the quantum state $\bm \varphi_4$:
\begin{equation}
\begin{aligned}
|\bm \varphi_4 \rangle =& \frac{1}{2}((|0\rangle + |1\rangle)|\bm \psi^{i+K}_{\operatorname{gt}} \rangle + (|0\rangle - |1\rangle)|\widetilde{\bm \psi}^{i+K}_{\operatorname{predict}} \rangle)\\
=&\frac{1}{2}(|0\rangle(| \bm \psi^{i+K}_{\operatorname{gt}} \rangle + |\widetilde{\bm \psi}^{i+K}_{\operatorname{predict}} \rangle) + |1\rangle(| \bm \psi^{i+K}_{\operatorname{gt}} \rangle - |\widetilde{\bm \psi}^{i+K}_{\operatorname{predict}} \rangle)).
\end{aligned}
\end{equation}

The probability of observing the ancillary qubit in the state $|1 \rangle$ can equivalently be represented using the expectation value of the Pauli-$Z$ operator acting on the ancillary qubit. The quantum state after applying the Hadamard gate is $|\bm \varphi_4\rangle$, and the expectation value of the Pauli-$Z$ operator on the ancillary qubit is given by:
\begin{equation}
    \langle Z \rangle = \langle \bm \varphi_4 | Z \otimes \mathbb{I} |\bm \varphi_4 \rangle,
\end{equation}
where $Z$ acts on the ancillary qubit, and $\mathbb{I}$ is the identity operator acting on the remaining qubits.

Substituting the quantum state $|\bm \varphi_4\rangle$, we write the expectation value of the Pauli-$Z$ operator is:
\begin{equation}
    \langle Z \rangle = \frac{1}{4} \Big( \langle \bm \psi^{i+K}_{\operatorname{gt}} + \widetilde{\bm \psi}^{i+K}_{\operatorname{predict}} | \bm \psi^{i+K}_{\operatorname{gt}} + \widetilde{\bm \psi}^{i+K}_{\operatorname{predict}} \rangle - \langle \bm \psi^{i+K}_{\operatorname{gt}} - \widetilde{\bm \psi}^{i+K}_{\operatorname{predict}} | \bm \psi^{i+K}_{\operatorname{gt}} - \widetilde{\bm \psi}^{i+K}_{\operatorname{predict}} \rangle \Big).
\end{equation}

Expanding the terms and simplifying, we find:
\begin{equation}
    \langle Z \rangle = \frac{1}{2} \Big( \langle \bm \psi^{i+K}_{\operatorname{gt}} |\widetilde{\bm \psi}^{i+K}_{\operatorname{predict}}  \rangle + \langle \widetilde{\bm \psi}^{i+K}_{\operatorname{predict}} | \bm \psi^{i+K}_{\operatorname{gt}}  \rangle \Big).
\end{equation}

The cost function $C(\bm \theta)$ is defined as:
\begin{equation}
    C(\bm \theta) = \langle \bm \psi^{i+K}_{\operatorname{gt}} - \widetilde{\bm \psi}^{i+K}_{\operatorname{predict}} | \bm \psi^{i+K}_{\operatorname{gt}} - \widetilde{\bm \psi}^{i+K}_{\operatorname{predict}} \rangle,
\end{equation}
which expands to:
\begin{equation}
    C(\bm \theta) = 2 \Big( \langle \bm \psi^{i+K}_{\operatorname{gt}} | \bm \psi^{i+K}_{\operatorname{gt}} \rangle + \langle \widetilde{\bm \psi}^{i+K}_{\operatorname{predict}} | \widetilde{\bm \psi}^{i+K}_{\operatorname{predict}} \rangle \Big) - 4 \langle Z \rangle.
\end{equation}

Rearranging terms, we have:
\begin{equation}
    C(\bm \theta) = 4 (1 - \langle Z \rangle).
\end{equation}

Thus, the cost function $C(\bm \theta)$ is directly related to the expectation value of the Pauli-$Z$ operator on the ancillary qubit.

\subsection*{B. Gradient calculations}
\label{GradientCalculations}
By employing equations \eqref{eq:UR} and \eqref{eq:mapping}, \eqref{eq:cost function} can be reformulated as:
\begin{equation}
    C(\bm \theta) = \langle \bm 0 |((U_R^{i+K})^\dagger- (U_R^{i})^\dagger[\bm F_k(\bm \theta)]^\dagger )(U_R^{i+K}-\bm F_k(\bm \theta)U_R^{i})|\bm 0 \rangle,
\end{equation}
where the symbol $\dagger$ denotes the conjugate transpose operation applied to a complex matrix. Subsequently, for parameters $\bm \theta$, the partial derivative can be expressed as:
\begin{equation}
\begin{aligned}
\frac{\partial C(\bm{\theta})}{\partial \theta_i} = & - \langle \bm{0} | (U_R^i)^\dagger \frac{\partial \bm{F}_k(\bm{\theta})^\dagger}{\partial \theta_i} (U_R^{i+K} - \bm{F}_k(\bm{\theta}) U_R^i) | \bm{0} \rangle \\
& - \langle \bm{0} | \big((U_R^{i+K})^\dagger - (U_R^i)^\dagger \bm{F}_k(\bm{\theta})^\dagger\big) \frac{\partial \bm{F}_k(\bm{\theta})}{\partial \theta_i} U_R^i | \bm{0} \rangle.
\label{eq:dCtheta}
\end{aligned}
\end{equation}

Gradient calculations in quantum computing are challenging due to the inability to directly measure quantum states. The parameter-shift rule proposed by Mitarai et al. (2018) \cite{mitarai2018quantum} offers an effective alternative, involving incremental parameter adjustments within a quantum circuit to estimate gradients. By calculating the differences in circuit outputs corresponding to small changes in each parameter, gradients can be estimated without the need for direct quantum state measurement, thus enabling gradient computations in quantum computing.

We employ the parameter-shift rule to compute the gradient $\partial \bm F_k (\bm \theta)/\partial \bm \theta$.  Through a series of deductions, \eqref{eq:dCtheta} can be expressed as:
\begin{equation}
\begin{aligned}
    \frac{\partial C(\bm{\theta})}{\partial \theta_i} &= \frac{1}{2}(\langle \bm 0 |({U_R^{j+K}}^\dagger- {\bm F_{i,+}(\bm \theta)}^\dagger{U_R^{j}}^\dagger)(U_R^{j+K}-U_R^{j}\bm F_{i,+}(\bm \theta))|\bm 0 \rangle 
    \\&- \langle \bm 0 |({U_R^{j+K}}^\dagger- {\bm F_{i,-}(\bm \theta)}^\dagger{U_R^{j}}^\dagger)(U_R^{j+K}-U_R^{j}\bm F_{i,-}(\bm \theta))|\bm 0 \rangle),
\end{aligned}  
\end{equation}
where we have defined:
\begin{equation}
    \bm F_{i,\pm}(\bm \theta) =  \bm F_{i,\pm}\left(\theta_1,\cdots,\theta_{i-1},\theta_{i}\pm \frac{\pi}{2},\theta_{i+1},\cdots,\theta_L\right).
\end{equation}

In this formulation, the ancillary qubit serves as a control mechanism for the gates that depend on $\bm \theta$. This enables efficient computation of the gradients for each parameter, as the parameter-shift rule evaluates the difference between circuit outputs for \(\theta_i \pm \frac{\pi}{2}\).

Finally, substituting \(\bm F_k(\bm \theta)\) into the gradient formula, the parameter-shift rule provides:
\begin{equation}
    \frac{\partial C(\bm{\theta})}{\partial \theta_k} = \frac{1}{2}\big[C(\bm \theta_{+k}) - C(\bm \theta_{-k})\big],
\end{equation}
where \(\bm \theta_{+k}\) and \(\bm \theta_{-k}\) denote the parameter sets with \(\theta_k\) shifted by \(\pm \frac{\pi}{2}\), respectively.

This approach avoids direct quantum state measurements by leveraging the ancillary qubit for controlled operations and enables efficient gradient-based optimization of the ansatz parameters within the quantum circuit.

\section*{Supplementary Note 4: 
Device information}

\begin{figure}
\centering
\includegraphics[width=17cm]{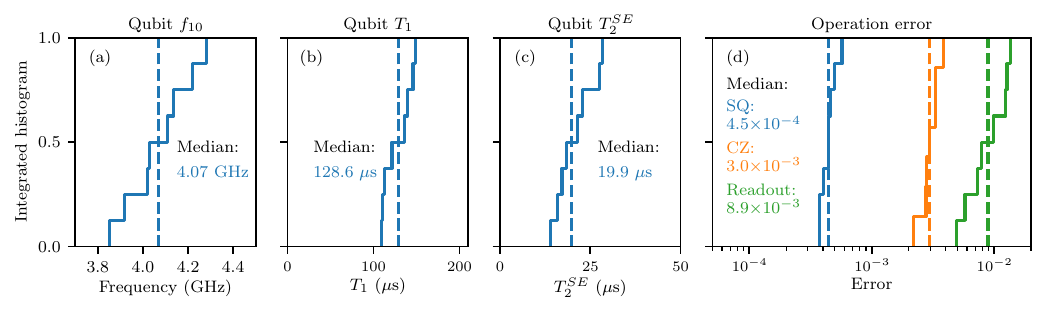} 
\caption{Integrated histograms of device performance parameters in the experiment. Dashed lines indicate the median values. (a) Qubit idle frequency $f_{10}$. (b) Qubit relaxation time $T_1$ measured at the idle frequency. (c) Qubit dephasing time $T^{SE}_2$ measured at the idle frequency with spin echo sequence. (d) Simultaneous operation errors of single-qubit gates (blue), two-qubit CZ gates (orange), and readout (green). The errors of quantum gates are Pauli errors obtained through simultaneous cross-entropy benchmarking, and the readout error is calculated as $e_r=1-(f_0+f_1)/2$ with $f_{0(1)}$ denoting the measure fidelities of the state $|0\rangle$ ($|1\rangle$).}
\label{fig:device}
\end{figure}

Our experiments are performed on a quantum processor featuring frequency-tunable transmon qubits and couplers, with the overall design detailed in Ref.~\cite{SM_Xu2023_digital}. 
The qubits are arranged in a $11\times 11$ square lattice, with adjacent qubits connected by tunable couplers. 
The effective coupling strength between two adjacent qubits can be tuned by biasing the coupler's frequency. 
The maximum resonance frequencies of qubits and couplers are around 4.5 GHz and 9.0 GHz, respectively.
Each qubit is capacitively couple to its own readout resonator for dispersively readout, with the frequency of the readout resonator being around 6.5 GHz.
We employ the gate set $\{U(\theta,\phi,\lambda), CZ\}$ to implement the experimental circuits, where $U(\theta,\phi,\lambda) = Rz(\theta)Ry(\phi)Rz(\lambda) $ denotes a generic single-qubit gate with three Euler angles, and CZ denotes the two-qubit controlled-Z gate that fits the layout topology of our device. Eight qubits on the processor are used in our experiments, with the characterized properties summarized in Fig.~\ref{fig:device}.

\newpage
\bibliographystyle{modified-apsrev4-2.bst}
\input{supp.bbl}

\end{document}

%% file: main.bbl
%

%% file: supp.bbl
%